\documentclass[12pt]{amsart}
\usepackage[margin=1in]{geometry}
\usepackage{amsfonts}
\usepackage{amssymb}
\usepackage[dvips]{graphics}
\usepackage{epsfig}
\usepackage{euscript}
\usepackage{color}
\usepackage{bbm}
\usepackage[colorlinks ,pdfstartview=FitB]{hyperref}
\hypersetup{allcolors=blue}

\newtheorem{thm}{Theorem}[section]
\newtheorem{cor}[thm]{Corollary}
\newtheorem{lem}[thm]{Lemma}

\newtheorem{ex}[thm]{Example}

\theoremstyle{definition}

\theoremstyle{remark}
\newtheorem{remark}[thm]{Remark}
\numberwithin{thm}{section}
\DeclareMathOperator{\RE}{Re}
\DeclareMathOperator{\IM}{Im}

\newcommand{\R}{{\mathord{\mathbb R}}}
\newcommand{\N}{{\mathord{\mathbb N}}}

\newcommand{\J}{\mathcal{J}}

\newcommand{\W}{\mathcal{W}}

\newcommand{\C}{{\mathord{\mathbb C}}}
\newcommand{\Z}{{\mathord{\mathbb Z}}}

\def\idty{{\mathchoice {\mathrm{1\mskip-4mu l}} {\mathrm{1\mskip-4mu l}} %
{\mathrm{1\mskip-4.5mu l}} {\mathrm{1\mskip-5mu l}}}}

\DeclareMathOperator{\Tr}{Tr}
\DeclareMathOperator{\Span}{span}

\DeclareMathOperator{\diag}{diag}
\DeclareMathOperator{\sgn}{sgn}
\begin{document}

\title[Entanglement of a class of non-Gaussian states]{Entanglement of a class of non-Gaussian states in disordered harmonic oscillator systems}

\author[H. Abdul-Rahman]{Houssam Abdul-Rahman}
\address{Department of Mathematics\\
University of Arizona\\
Tucson, AZ 85721, USA}
\email{houssam@math.arizona.edu}
\date{\today}

\begin{abstract}
For disordered harmonic oscillator systems over the $d$-dimensional lattice, we consider the problem of finding the bipartite entanglement of the uniform ensemble of the energy eigenstates associated with a particular number of modes. Such ensemble define a class of mixed, non-Gaussian entangled states that are labeled, by the energy of the system, in an increasing order.  We develop a novel approach to find the exact logarithmic negativity of this class of states. We also prove entanglement bounds and demonstrate that the low energy states follow an area law.
\end{abstract}

\maketitle

%%%%%
%
%  Intro . . .
%
%%%%%%%

\allowdisplaybreaks
\section{Introduction}

Quantum entanglement  is considered as a key resource for many quantum technologies and information processing \cite{Nielsen2000, Horodecki09}. Gaussian states are profoundly important  in quantum information \cite{Wangetal08,Weedbrook12}; however, the entanglement phenomenon is not confined to this well studied set of states.   It has recently become apparent that there is a fundamental need to understand the entanglement beyond Gaussian states \cite{Gomesetal09,Strobeletal14,Walschaerseta17}, as
non-Gaussian states are not only advantageous but sometimes necessary to perform certain quantum information tasks; for instance, non-Gaussian states are required for entanglement distillation \cite{Dong08,Eisert02,Hage08}; they can be used as key ingredients to improve quantum teleportation
\cite{Dell07,Opartny04,Kaushik15}; and they are necessary for universal quantum computation with continuous variables systems  \cite{Bartlett02}.

On the theoretical level, the study of entanglement of many-body systems has been limited to Gaussian states, see e.g., \cite{ANSS16,EisertCramerPlenio2010,LamiAdessoetal17, NSS2}, where the problem of finding entanglement reduces to the study of the corresponding correlation (covariance) matrices. The latter are well structured matrices of dimension equal to double the volume of the many-body system. The entanglement story of non-Gaussian states has no such happy ending. General states of many-body systems are usually extremely complicated objects with a number of parameters exponentially large in the volume of the system.  For this reason, there are almost no rigorous results on the formidable task of understanding the entanglement of non-Gaussian states.

Understanding energy eigenstates is one of the main problems in many-body and statistical physics. In particular, for disordered regimes, they are of utmost importance in the context of \textit{many-body localization} (MBL), a phenomenon that has recently received strong attention in the physical community, see e.g., \cite{AbanicPepic17,Agarwaletal,Altmam-etal-15,NH} for recent reviews with extensive lists of references. Proving rapid \textit{decay of correlations} or \textit{area laws} for eigenstates (at least for the low lying eigenstates) is believed to be a signature of the MBL phase. Examples of mathematical results concerning eigenstate localization beyond the ground state include the XY spin chain in random transversal field (see \cite{ANSS17, ARS15} for area laws and \cite{SimsWarzel} for decay of correlations.) and the Tonks-Girardeau gas \cite{SW}. Recently, \cite{EKS} proved exponential clustering of all eigenstates throughout the droplet spectrum of the XXZ chain in a random field and \cite{BW18} proved area laws for the states in the droplet regime. In fact, \cite{BW18} represents one of the very few rigorous results on entanglement for non-Gaussian states.

The system of coupled harmonic oscillators is unitary equivalent to a free boson system, reducing it to the diagonalization of an effective one-particle Hamiltonian.
 The ground state and thermal states of harmonic oscillators are two simple examples of Gaussian states with a long history of entanglement results, see e.g.,  \cite{Audenaert2002,SCW} for the deterministic gapped case and \cite{NSS2} for the disordered gapless case. It has been shown that the entanglement of these Gaussian states follows area laws. But unlike the free fermion systems case, the eigenstates of the coupled harmonic oscillators beyond the ground state are non-Gaussian, and thus understanding their entanglement is an existing challenge. Recently, \cite{ARSS17} proved exponential clustering results of arbitrary eigenstates for disordered harmonic oscillator; the corresponding bound increases with the highest mode, which indicates area laws for the low lying eigenstates \cite{BH13,BH}.

In this paper, we consider the problem of finding the entanglement of a class of mixed, non-Gaussian states of (disordered) harmonic oscillator systems. These are states written as a uniform ensemble of energy eigenstates associated with a particular number of modes $N$ (the $N$-modes ensemble states). The expected energy of the system in this class of non-Gaussian states increases with the total number of modes, which means that they can be labeled, in an explicit way, in an increasing order with respect to the number of modes.  As expected, the resulting entanglement formula for the $N$-modes ensemble states is more complicated than the one associated with Gaussian states, but it is fully characterized by the static ground state correlation matrix. The latter is known to decay exponentially, which allows to find an entanglement bound that is proportional to the surface area of the boundary region between the two subregions, but grows linearly in the total number of modes (see Theorem \ref{thm:AreaLaw}). Such estimation is understood as a  genuine area law for low particle concentration (relative to the volume of the system), and it demonstrates MBL in the so-called \textit{zero temperature regime}, as it indicates that the low lying eigenstates are weakly entangled.

We build on the ideas by Vidal and Werner \cite{VidalWerner} and on the arguments presented in \cite{NSS2} to develop the machinery needed  to find the exact formula for the logarithmic negativity of the $N$-modes ensemble states. Theorem \ref{thm:Weyl-Eigen} provides general formulas for the expectation of the Weyl operators at arbitrary eigenstates. These formulas include products of multi-variable Laguerre polynomials as pre-factors for the gaussian term, which reflects the fact that the corresponding states are non-Gaussian. Lemma \ref{lem:I-N} represents the central technical fact about multi-variable Laguerre polynomials. It distinguishes the $N$-modes ensemble states from general eigenstates, it says that the ensemble averaged multi-variable Laguerre polynomials is a (generalized Laguerre) polynomial of the \textit{sum of all variables}.

The paper is organized as follows. In the next section, we describe the model, introduce our class of states, and present the main results. Remark \ref{rem:det} comments on the validity of our results for a class of deterministic gapped oscillator models.
 Section \ref{sec:Weyl} includes new results regarding  the Weyl operator expectations. In Section \ref{sec:PT:diag} we present the complete diagonalization of the partial transpose as needed to calculate the logarithmic negativity. The area laws are then proven in Section \ref{sec:LN:bound}.

\section*{Acknowledgment}
The author would like to thank Bruno Nachtergaele, Robert Sims, and G\"{u}nter Stolz for insightful  discussions and comments.

\section{The model and results}\label{sec:Model}
\subsection{The Model}
For integers $a<b$ and $d\geq 1$, let $\Lambda$ be the $d$-dimensional box $\Lambda=\Z^d\cap[a,b]^d$, equipped with the $\ell^1$-metric, where we use the notation $|\cdot|$ to denote the 1-norm. We consider harmonic oscillator systems defined over $\Lambda$, coupled by nearest neighbor quadratic interactions, quantified by the interaction's strength parameter $\lambda>0$. The Hamiltonian is
\begin{equation}\label{eq:def:H}
H_\Lambda=\sum_{x\in\Lambda}(p_x^2+k_x q_x^2)+
\sum_{{\tiny\begin{array}{c}\{x,y\}\subset\Lambda\\|x-y|=1\end{array}}}\lambda(q_x-q_y)^2,
\end{equation}
acting on the Hilbert space
\begin{equation}
\mathcal{H}_\Lambda=\bigotimes_{x\in\Lambda}
\mathcal{L}^2(\R)\cong\mathcal{L}^2(\R^\Lambda),
\end{equation}
where for each site $x\in\Lambda$, we denote the position operator by  $q_x$, i.e., the multiplication operator by $q_x$  in $\mathcal{H}_\Lambda$, and by $p_x=-i\partial/\partial q_x$, we denote the corresponding momentum operator. It is well know, see e.g., \cite{ReedSimon}, that the position and the momentum operators are unbounded self-adjoint on suitable domains and satisfy the commutation relations
\begin{equation}\label{eq:def:pq}
[q_x,q_y]=[p_x,p_y]=0\ \quad\text{and}\quad [q_x,p_y]=i\delta_{x,y}\idty
\quad\text{ for all }x,y\in\Lambda,\end{equation}
where $\delta_{x,y}$ is the Kronecker delta function. For each $x\in\Lambda$,  $k_x$ is the corresponding spring constant.  We regard the elements of the sequence $\{k_x\}_{x\in\Lambda}$ as i.i.d. random variables with absolutely continuous distribution given by bounded density, supported in $[0,k_{\max}]$ with $0<k_{\max}<\infty$. The coupling sum in (\ref{eq:def:H}) is taken over all undirected edges $\{x,y\}$ from $\Lambda$ that correspond to the nearest neighbor sites $x$ and $y$, i.e., $|x-y|=1$.

As is well known, the analysis of the Hamiltonian $H_\Lambda$ reduces to the study of the finite volume Anderson model
\begin{equation}\label{eq:def:h}
h_\Lambda=\lambda h_{0,\Lambda}+k,
\end{equation}
where $h_{0,\Lambda}$ is the negative discrete Laplacian on $\Lambda$, and $k=\diag\{k_x,\, x\in\Lambda\}$. In particular,  the harmonic oscillators Hamiltonian $H_\Lambda$ can be written as the bosonic system
\begin{equation}
H_\Lambda=\frac{1}{2}\begin{pmatrix}
            a^T & (a^*)^T \\
          \end{pmatrix}\begin{pmatrix}
                         h_\Lambda-\idty & h_\Lambda+\idty \\
                         h_\Lambda+\idty & h_\Lambda-\idty \\
                       \end{pmatrix}\begin{pmatrix}
                                      a \\
                                      a^* \\
                                    \end{pmatrix}.
\end{equation}
where $\begin{pmatrix}
        a^T & (a^*)^T \\
      \end{pmatrix}
$ and $\begin{pmatrix}
        a \\
        a^* \\
      \end{pmatrix}$ are the $2|\Lambda|$ row and column vectors of the \textit{annihilation} and \textit{creation} operators
\begin{equation}\label{eq:def:a}
a_x=\frac{1}{\sqrt{2}}(q_x+ip_x),\quad a_x^*=\frac{1}{\sqrt{2}}(q_x-ip_x),\text{\,for\,}x\in\Lambda.
\end{equation}
As a direct consequence of (\ref{eq:def:pq}) and (\ref{eq:def:a}), the operators $\{a_x\}_{x\in\Lambda}$ satisfy the \textit{canonical commutation relations} (CCR)
\begin{equation}\label{eq:CCR}
[a_x,a_y]=[a_x^*,a_y^*]=0, \text{\ and\ } [a_x,a_y^*]=\delta_{x,y}\idty, \text{\, for all \,} x,y\in\Lambda.
\end{equation}

It is clear that $h_\Lambda$ is self-adjoint and since we work in finite volume, the spectrum of $h_\Lambda$ is discrete and, under our assumption of absolutely continuous distribution of the $\{k_x\}_{x\in\Lambda}$, almost surely simple. Moreover, using that $0\leq h_{0,\Lambda}\leq 4d$, the spectrum of $h_\Lambda$ is included within the interval
\begin{equation}\label{eq:def:spectrumh}
\sigma(h_\Lambda)\subseteq\left[\min_{x\in\Lambda}k_x,(4d\lambda+k_{\max})\right].
\end{equation}
Meaning that $h_\Lambda$ is almost surely positive and  the operator $h_\Lambda^{-\frac{1}{2}}$ is not uniformly bounded in the volume of the system and the disorder.

By standard results, see  Proposition A.3(c) in \cite{NSS1}, the elements of $h_\Lambda^{-1/2}$ decay exponentially away from the diagonal after averaging the disorder, there exist constants $C<\infty$ and $\mu>0$ such that,
\begin{equation}\label{eq:ecl}
\mathbb{E}\left(|\langle\delta_x,h_\Lambda^{-1/2}\delta_y\rangle|\right)
\leq C e^{-\mu|x-y|}
\end{equation}
for all $x,y\in\Lambda$, where $\{\delta_x,x\in\Lambda\}$ is the canonical basis of $\ell^2(\Lambda)$ and $\mathbb{E}(\cdot)$ is the disorder average. We remark here that (\ref{eq:ecl}) leads to the exponential decay of the static position-position correlations at the ground state of $H_\Lambda$, see e.g., \cite{ARSS17,NSS1}.

The real nonnegative operator $h_\Lambda$ can be diagonalized in terms of an orthogonal matrix $\mathcal{O}$, i.e., $
      h_\Lambda=\mathcal{O}\gamma^2\mathcal{O}^T$
      where $\gamma^2=\diag\{\gamma_x^2,\ x\in\Lambda\}$. This defines the Bogoliubov transformation
      \begin{equation}\label{eq:a-to-b}
      b=\frac{1}{2}(\gamma^\frac{1}{2}+\gamma^{-\frac{1}{2}})\mathcal{O}^T\ a+ \frac{1}{2}(\gamma^\frac{1}{2}-\gamma^{-\frac{1}{2}})\mathcal{O}^T\ a^*,
      \end{equation}
 which leads to the second quantization representation of the harmonic oscillator systems $H_\Lambda$
\begin{equation}\label{eq:def:H-bosons}
H_\Lambda=\sum_{k=1}^{|\Lambda|}\gamma_k(2b_k^* b_k+\idty).
\end{equation}
For more details regarding this reduction to a free boson system we refer the reader to e.g., \cite{NSS1,NSS2,ARSS17}.
Due to the Bogoliubov transformation (\ref{eq:a-to-b}), the operators $\{b_k\}_k$  satisfy the CCR. These modes (or quasi-particles) are mixing the creation and annihilation operators.
The complete set of eigenfunctions of $H_\Lambda$ consists of the Fock states generated from the vacuum $\psi_0$ of the $b$ operators ($b_x\psi_0=0$ for all $x\in\Lambda$.)
\begin{equation}\label{eq:def:e-fun-bosons}
\psi_{\alpha}=\prod_{k=1}^{|\Lambda|}
\frac{1}{\sqrt{\alpha_k}}(b_k^*)^{\alpha_k}
\psi_0,
\end{equation}
for any occupation number vector $\alpha=(\alpha_1,\ldots,\alpha_{|\Lambda|})\in\N^{|\Lambda|}_0$.
It follows directly that
\begin{equation}\label{eq:H:eigenvalues}
H_\Lambda\psi_\alpha=\left(\sum_{k=1}^{|\Lambda|}\gamma_{k}(2\alpha_k+1)\right)\,\psi_\alpha.
\end{equation}
We note here that beside the total number of modes, the distribution of modes over $\Lambda$ plays a crucial role in the determination of the energy at a given eigenstate; yielding no explicit way, in terms of $\alpha$, to label all energy eigenstates in order of increasing energy levels. As is clear from (\ref{eq:H:eigenvalues}), the ground state gap is $2\min_{k}\gamma_k$.
\begin{remark}\label{rem:det}
 The results we proved below are also valid for a class of  the deterministic, gapped oscillator models over the $d$-dimensional box $\Lambda$
\begin{equation}
\tilde{H}_\Lambda=\sum_{x\in\Lambda}p_x^2+\sum_{x,y\in\Lambda} q_x h^{(q)}_{xy} q_y
\end{equation}
defined over the Hilbert space $\mathcal{H}_\Lambda$, where $h^{(q)}_\Lambda:=(h^{(q)}_{xy})_{x,y\in\Lambda}$ is real, symmetric, and positive definite banded matrix. We require that $h^{(q)}_\Lambda$ is bounded and boundedly invertible, uniformly in the volume $|\Lambda|$; that is, there exist constants $0<C_1<C_2<\infty$ independent of $\Lambda$ such that
\begin{equation}\label{eq:hq:bounds}
C_1\leq\|h^{(q)}_\Lambda\|\leq C_2.
\end{equation}
We assume that the spectrum of $h^{(q)}_\Lambda$ is simple. Here, $h^{(q)}_\Lambda$ substitutes $h_\Lambda$, given in (\ref{eq:def:h}), in the disordered regime described above. The analogue of the exponential decay (\ref{eq:ecl}) is satisfied, of course with no need to average, see e.g., \cite{Benzietal15}. Moreover, it follows from (\ref{eq:H:eigenvalues}) and (\ref{eq:hq:bounds}) that $\tilde{H}_\Lambda$ is a gapped system, i.e., there is a robust ground state gap.
\end{remark}

\subsection{Results}
We study the bipartite entanglement of a class of states describing the harmonic oscillator Hamiltonian $H_\Lambda$. Explicitly, we fix a subregion $\Lambda_0\subset\Lambda$ and consider the bipartite decomposition $\mathcal{H}_\Lambda=\mathcal{H}_{\Lambda_0}\otimes\mathcal{H}_{\Lambda\setminus\Lambda_0}$, where
\begin{equation}\label{eq:decomposition:H}
\mathcal{H}_{\Lambda_0}=\bigotimes_{x\in\Lambda_0}\mathcal{L}^2(\R)\quad\text{ and }\quad \mathcal{H}_{\Lambda\setminus\Lambda_0}=\bigotimes_{x\in\Lambda\setminus\Lambda_0}\mathcal{L}^2(\R).
\end{equation}
For any state $\rho\in\mathcal{B}(\mathcal{H}_\Lambda)$, denote by $\rho^{T_1}$, the \textit{partial transpose} of $\rho$ with respect to the first component in the decomposition (\ref{eq:decomposition:H}), and the \textit{logarithmic negativity} $\mathcal{N}(\rho)$ of $\rho$ is the logarithm of the trace norm of the state's partial transpose
\begin{equation}
\mathcal{N}(\rho)=\log\|\rho^{T_1}\|_1.
\end{equation}
It is well known that the logarithmic negativity is an upper bound to the distillable entanglement \cite{VidalWerner}.
We remark that since the bosons $\{b_k\}_k$ are mixing the spacial creation and annihilation operators as in (\ref{eq:a-to-b}), the eigenstates $\psi_\alpha$ are entangled with respect to the decomposition (\ref{eq:decomposition:H}).
The problem of finding their entanglement bumps into the fact that the ground state is the only Gaussian eigenstate of $H_\Lambda$, see the beginning of Section \ref{sec:Weyl} for a brief definition of Gaussian states on the CCR, references for more details are also provided.

In this work, we find the logarithmic negativity of a class of mixed, entangled, and non-Gaussian states. Explicitly, these states are defined as follows: for each $N\in\mathbb{N}$, we consider the set $\J_N$ that consists of all occupation number vectors that are describing a total of $N$ modes,
\begin{equation}\label{eq:def:Jn}
\J_N:=\{\alpha=(\alpha_1,\ldots,\alpha_{|\Lambda|})\in\mathbb{N}_0^{|\Lambda|}; \ \sum_j\alpha_j=N\},
\end{equation}
then the state under consideration is the uniform statistical ensemble of the $N$-modes eigenstates
\begin{equation}\label{eq:def:rhoN}
\rho_N:=\frac{1}{|\J_N|}\sum_{\alpha\in \J_N}|\psi_\alpha\rangle\langle\psi_\alpha|.
\end{equation}
Note that the $N$-modes ensemble $\rho_N$ is the orthogonal projection onto the Fock space sector $\mathcal{F}_N$ characterized by  $N$ modes, $\mathcal{F}_N=\Span\{\psi_\alpha;\ \sum\alpha_j=N\}$.
Moreover, $\rho_N$ is the maximum entropy state of the total (quasi-)particle number operator $N^{(b)}:=\sum_{j=1}^{|\Lambda|}b_j^* b_j$, see e.g., \cite{MES-Eisert}.

It is remarkable that, in contrast to the case with the eigenstates, the expected energy of the harmonic system at $\rho_N$ is increasing in the total number of modes $N$, see Lemma \ref{sec:App:Energy},
\begin{equation}\label{eq:H:rhoN:energy}
\Tr[ H_\Lambda \rho_N]=\left(1+\frac{2 N}{|\Lambda|}\right)\sum_{k}\gamma_k.
\end{equation}
Thus, the $N$-modes ensemble states can be used to study the energy dependence of their entanglement. This is
beside the fact that the entanglement of $\rho_N$ is made, throughout this work, theoretically accessible. As understood intuitively, it follows from (\ref{eq:H:rhoN:energy}) that the energy of the system goes to the ground state energy in the thermodynamic limit. Note also that the expected energy of $H_\Lambda$ at the $N$-modes ensemble with $N=K|\Lambda|$ modes matches the energy of the system at the eigenstate that correspond to the occupation number vector $(K,K,\ldots,K)$.

The $N$-modes ensemble $\rho_N$ is a mixed state by construction, it is obviously entangled since all the energy eigenstates of $H_\Lambda$ are entangled, and aside from the degenerate case associated with the ground state $\rho_{N=0}=|\psi_0\rangle\langle\psi_0|$, $\rho_N$ is not Gaussian.

Our first new result is the exact evaluation of the logarithmic negativity of the $N$-modes ensemble $\rho_N$, which is done through the complete diagonalization of the partial transpose of $\rho_N$ in Theorem \ref{thm:rhoT1:eigen} below. This diagonalization is made possible using the ideas in Theorems \ref{thm:Weyl-Eigen} and \ref{thm:rho1}. The former theorem gives a formula for the Weyl operator expectations at arbitrary eigenstates, then Theorem \ref{thm:rho1} defines an operator that corresponds to a scaled version  of the Weyl operator expectations.

The second result provides a growing area law in the total number of modes. Meaning that, on average, the logarithmic negativity of $\rho_N$ with respect to the decomposition (\ref{eq:decomposition:H}) is proportional to the surface area between the two subsystem, with a coefficient that grows linearly in $N$.
Let's denote the boundary of $\Lambda_0$ by $\partial\Lambda_0$,
\begin{equation}
\partial\Lambda_0=\{x\in\Lambda_0;\, \exists y\in\Lambda\setminus\Lambda_0 \text{\,with\,}|x-y|=1\}.
\end{equation}
 We will prove the following theorem in Section \ref{sec:LN:bound}.
\begin{thm}\label{thm:AreaLaw}
For any $\Lambda_0\subset\Lambda$, $N\in\mathbb{N}_0$, and the corresponding N-modes ensemble $\rho_N$ as in (\ref{eq:def:rhoN}), there exists $\tilde{C}<\infty$ such that
\begin{equation}\label{eq:AreaLaw}
\mathbb{E}\left(\mathcal{N}(\rho_N)\right)\leq \tilde{C}(2N+1)|\partial\Lambda_0|
\end{equation}
where the constant $\tilde{C}$ is independent of $N$, $\Lambda_0$ and $\Lambda$, but it depends on the parameters $C$ and $\mu$ from (\ref{eq:ecl}) as well as  on $\lambda$, $k_{\max}$ and $d$.
\end{thm}
This is a typical area law for the low lying non-Gaussian states $\rho_N$, that is, when $N/|\Lambda|\ll 1$.
The proof follows from the bound in Lemma \ref{lem:LN:upper1} below and the exponential clustering of the ground state, a consequence of  (\ref{eq:ecl}). This has to mean that exponential clustering of the ground state implies an area law for the ground state and for the low lying states close to the ground state.
%\begin{remark}
%We note that the entanglement bound %(\ref{eq:AreaLaw}) is in the spirit of the %bound that one can get for the %positions-momenta correlations at $\rho_N$. %This can be seen as follows, Theorem 1.2 from %\cite{ARSS17} proves exponential clustering of %eigenstates, with a bound that increases with %the highest mode. More precisely, there exist %$C'<\infty$ and $\eta>0$ such that
%\begin{equation}
%\mathbb{E}\left(|\langle p_x %p_y\rangle_{\alpha}|\right)\leq %C'(1+\|\alpha\|_\infty)^2e^{-\eta|x-y|}
%\end{equation} for all $x,y\in\Lambda$,
%where $\langle p_x %p_y\rangle_\alpha=\Tr\left[p_x %p_y\,|\psi_\alpha\rangle
%\langle\psi_\alpha|\right]$ is the %$\psi_\alpha$ expectation of $p_x p_y$. This %gives
%\begin{equation}
%\mathbb{E}\left(\langle p_x %p_y\rangle_{\rho_N}\right)\leq %\frac{1}{|\J_N|}\sum_{\alpha\in\J_N}C'(1+
%\|\alpha\|_\infty)^2e^{-\eta|x-y|}\leq %C'(N+1)^2e^{-\eta|x-y|}.
%\end{equation}
%\end{remark}

A variant version of the proof of Theorem \ref{thm:AreaLaw} provides similar entanglement bounds for a larger class  of mixed non-Gaussian states. We construct them as follows: for any $N\in\mathbb{N}_0$ let
\begin{equation}\label{eq:def:Omega}
\Omega_N:=\{\omega=(\omega_0,\omega_1,\ldots,\omega_N)\in[0,1]^{N+1},\ \sum_j\omega_j=1\}
\end{equation}
then for each $\omega\in\Omega_N$ we can define a mixed state that describes the weighted average of the $L$-modes ensemble states for $L=0,\ldots,N$
\begin{equation}\label{eq:def:tilderhoN}
\rho_{\omega,N}:=\sum_{L=0}^{N}\omega_L\rho_L,
\end{equation}
where $\rho_L$ is as in (\ref{eq:def:rhoN}). It follows from (\ref{eq:H:rhoN:energy}) that for each $N\in\mathbb{N}$, the maximum expected energy at $\rho_{\omega,N}$ is the expected energy of the system at the $N$-modes ensemble, i.e.,
$
\max_{\omega\in\Omega_N}\langle H_\Lambda\rangle_{\rho_{\omega,N}}=\langle H_\Lambda\rangle_{\rho_N}$. This makes this class of states less interesting than the $N$-modes states from the energy-entanglement perspective. But one can still argue that we are providing entanglement bounds for a larger class of non-Gaussian states. The following corollary says that the logarithmic negativity of $\rho_{\omega,N}$ has the same bound as $\mathcal{N}(\rho_N)$.
\begin{cor}\label{cor:AreaLaw2}
For any $N\in\mathbb{N}_0$ and $\rho_{\omega,N}$ as in (\ref{eq:def:tilderhoN}) with $\omega\in\Omega_N$ as in (\ref{eq:def:Omega}), then with the same $\tilde{C}>0$ as in (\ref{eq:AreaLaw}) we have
\begin{equation}\label{eq:AreaLaw2}
\mathbb{E}\left(\max_{\omega\in\Omega_N}\
\mathcal{N}(\rho_{\omega,N})\right)\leq \tilde{C}(2N+1)|\partial\Lambda_0|.
\end{equation}
\end{cor}
Note that the logarithmic negativity is not convex, meaning that Corollary \ref{cor:AreaLaw2} is not just a direct consequence of Theorem \ref{thm:AreaLaw}. We include the proof of Corollary \ref{cor:AreaLaw2} at the end of Section \ref{sec:LN:bound}.

\section{The Weyl operator expectations}\label{sec:Weyl}
 As was the case in all exact entanglement results of Gaussian states, e.g.,
\cite{VidalWerner,NSS2,Audenaert2002,EisertCramerPlenio2010}, our approach to the entanglement of the class $\{\rho_N,N\in\mathbb{N}\}$ of non-Gaussian states starts from the Weyl operator expectations, often called the \textit{(quantum) characteristic functions}. We first define the Weyl operators then prove Theorem \ref{thm:Weyl-Eigen}, that provides a formula for the characteristic function of arbitrary eigenstates of the harmonic oscillators Hamiltonian $H_\Lambda$. Theorem \ref{thm:rho1} provides a stepping stone to study the entanglement of the class of $N$-modes ensemble $\rho_N$.

For every $f:\Lambda\rightarrow \mathbb{C}$, the \textit{Weyl} (or \textit{displacement}) operator is defined as the unitary operator
\begin{equation}
\W(f)=\exp{\left(\frac{i}{\sqrt{2}}(a(f)+a^*(f))\right)}\text{ where } a(f)=\sum_{x\in\Lambda} \overline{f_x}a_x \quad a^*(f)=\sum_{x\in\Lambda} f_x a_x^*,
\end{equation}
where $a_x$ and $a_x^*$ are the annihilation and creation operators defined in (\ref{eq:def:a}). The CCR gives directly that
\begin{equation}\label{eq:Weyl:tensor}
\W(f)=\bigotimes_{x\in\Lambda}\W_{f_x},\quad \text{ where } \W_{f_x}:=\exp\left(\frac{i}{\sqrt{2}}(\overline{f_x}a_x+f_x a^*_x)\right).
\end{equation}
Let's identify $\ell^2(\Lambda;\mathbb{C})$ with $\ell^2(\Lambda;\mathbb{R})\oplus\ell^2(\Lambda;\mathbb{R})$, i.e.,
\begin{equation}\label{eq:def:tildef}
f\in\ell^2(\Lambda;\mathbb{C})\ \sim\ \tilde{f}=\begin{pmatrix}
            \RE f \\
            \IM f \\
          \end{pmatrix}\in
          \ell^2(\Lambda;\mathbb{R})\oplus\ell^2(\Lambda;\mathbb{R}).
\end{equation}
A state $\rho$ on $\mathcal{B}(\mathcal{H}_\Lambda)$ with vanishing first moments, i.e., $\langle q_x \rangle_\rho=\langle p_x\rangle_\rho=0$ for all $x\in\Lambda$ (like the eigenstates and thermal states of free boson systems), and with position-momentum correlation matrix $\Gamma_\rho=\langle R R^T
\rangle_{\rho}$ with
$R=
\begin{pmatrix}
  q \\
  p \\
\end{pmatrix}
$ is said to be \textit{Gaussian} (or \textit{quasi-free}) if it satisfies the bosonic Wick's theorem, and it follows that such Gaussian states are exactly those with gaussian  characteristic function,
\begin{equation}\label{eq:Weyl:QF}
\langle\W(f)\rangle_\rho:=\Tr[\W(f)\rho]=e^{-\frac{1}{2}\langle\tilde{f},\Gamma_\rho\tilde{f}\rangle},
\end{equation}
for all $f\in\ell^2(\Lambda)$.
For a definition of general Gaussian states on the CCR algebra we refer the reader to e.g., \cite{BRvol2,Man68,Araki71,NSS2}.

\subsection{The characteristic function of eigenstates}
Theorem \ref{thm:Weyl-Eigen} gives a formula for the characteristic function of any eigenstate $\rho_\alpha=|\psi_\alpha\rangle\langle\psi_\alpha|$ of $H_\Lambda$ in terms of the ground state correlation matrix $\Gamma_0$. To state the theorem, we first introduce the $2|\Lambda|\times 2|\Lambda|$ matrix $M$
\begin{equation}\label{eq:def:M}
M=\begin{pmatrix}
    h_\Lambda^{-\frac{1}{2}} & 0 \\
    0 & h_\Lambda^{\frac{1}{2}} \\
  \end{pmatrix}
\end{equation}
where $h_\Lambda=\mathcal{O}\gamma^2\mathcal{O}^T$ is defined in (\ref{eq:def:h}). For each $k=1,\ldots,|\Lambda|$, let $\chi_k(M)$ be the orthogonal  projection of $M$ onto a subspace spanned by the eigenvectors associated with the eigenvalues $\gamma_k^{-1}$ and $\gamma_k$. i.e., $\chi_k(M)$ is defined as
\begin{equation}\label{eq:chikM}
\chi_k(M)=\begin{pmatrix}
            \mathcal{O} & 0 \\
            0 & \mathcal{O} \\
          \end{pmatrix}\chi_{k}\begin{pmatrix}
            \mathcal{O}^T & 0 \\
            0 & \mathcal{O}^T \\
          \end{pmatrix},
\end{equation}
where $\chi_k$ is the restriction operator onto the $k$-th site,
\begin{equation}
\chi_k:=\begin{pmatrix}
          \idty_{\{k\}} & 0 \\
          0 & \idty_{\{k\}}  \\
        \end{pmatrix}, \text{ and }\idty_{\{k\}}:=e_ke_k^T.
\end{equation}

We remark that $M=2\Gamma_0+iJ$ where $J=\begin{pmatrix}
                                                  0 & -\idty \\
                                                  \idty & 0 \\
                                                \end{pmatrix}
$ and $\Gamma_0$ is the static position-momentum correlation matrix at the ground state, see e.g., \cite{ARSS17}.

Now we are ready to state the theorem.
\begin{thm}\label{thm:Weyl-Eigen} Let $\alpha=(\alpha_1,\ldots,\alpha_{|\Lambda|}) \in \mathbb{N}_0^{| \Lambda|}$ be the vector of occupation modes and let $\rho_{\alpha}=|\psi_\alpha\rangle\langle\psi_\alpha|$ be the corresponding eigenstate of $H_\Lambda$. Then for any $f: \Lambda \to \mathbb{C}$, the characteristic function of $\rho_\alpha$ is given by the formula
\begin{equation} \label{weyl_corM}
\langle \W(f) \rangle_{\rho_\alpha} =e^{-\frac{1}{4}\langle\tilde{f},M\tilde{f}\rangle} \prod_{k=1}^{| \Lambda|} L_{\alpha_k} \left(\frac{\langle\tilde{f},M\chi_k(M)\tilde{f}\rangle}{2} \right).
\end{equation}
Here $\tilde{f}$ is as in (\ref{eq:def:tildef}), $M$ is given in (\ref{eq:def:M}), $\chi_k(M)$ is its spectral projection given in (\ref{eq:chikM}), and $L_{\alpha_k}(\cdot)$ is the Laguerre polynomial of degree $\alpha_k$, defined as
\begin{equation}\label{eq:def:Laguerre}
L_{k}(x)=\sum_{n=0}^{k}\frac{(-1)^n x^n}{n!}\dbinom{k}{n}, \text{ for }k=1,2,\ldots.
\end{equation}
\end{thm}
In the ground state $\rho_0$, the statement (\ref{weyl_corM}) reduces to the well known formula
\begin{equation}
\langle\W(f)\rangle_{\rho_0}=e^{-\frac{1}{4}\langle \tilde{f},M\tilde{f}\rangle},
\end{equation}
see e.g., \cite{NSS1} which originates from Proposition 5.2.28 of \cite{BRvol2} or Chapter XII.12 of \cite{Messiah}.
In comparison with the characteristic function of Gaussian states in (\ref{eq:Weyl:QF}), and using that $M=2\Gamma_{0}+iJ$, with the fact $\langle \tilde{f},J\tilde{f}\rangle=0$,  we observe that the ground state of $H_\Lambda$ is the only Gaussian eigenstate.
\begin{proof}
We first show that the characteristic function formula (\ref{weyl_corM}) is equivalent to
\begin{equation}\label{eq:Weyl:alpha}
\langle \W(f)\rangle_{\rho_\alpha}=e^{-\frac{1}{4}\|Vf\|^2}\prod_{k=1}^{|\Lambda|}L_{\alpha_k}\left(\frac{|(Vf)_k|^2}{2}\right),
\text{ where }Vf=\gamma^{-\frac{1}{2}}\mathcal{O}^T \RE[f]+i\gamma^{\frac{1}{2}}\mathcal{O}^T \IM[f].
\end{equation}
Formula (\ref{eq:Weyl:alpha}) follows from the following argument: first, define $\tilde{V}$ as
\begin{equation}
\tilde{V}=\begin{pmatrix}
    \gamma^{-\frac{1}{2}}\mathcal{O}^T & 0 \\
   0 &  \gamma^{\frac{1}{2}}\mathcal{O}^T\\
  \end{pmatrix}
\end{equation}
and note that $\tilde{V}^T\tilde{V}=M$.
Then it is straight forward to check the following two statement
\begin{eqnarray}
\|Vf\|^2&=&\langle Vf,Vf\rangle_{\ell^2(\Lambda;\mathbb{C})}=\langle \tilde{V}\tilde{f},\tilde{V}\tilde{f}\rangle_{\ell^2(\Lambda;\mathbb{R})\oplus\ell^2(\Lambda;\mathbb{R})}
=\langle\tilde{f},M\tilde{f}\rangle.\\
|(Vf)_k|^2&=&\|\chi_k\tilde{V}\tilde{f}\|^2_{\ell^2(\Lambda;\mathbb{R})\oplus\ell^2(\Lambda;\mathbb{R})}
=\langle\tilde{f},\tilde{V}^T\chi_k\tilde{V}\tilde{f}\rangle
=\langle\tilde{f},M\chi_k(M)\tilde{f}\rangle.
\end{eqnarray}
So we proceed by proving formula (\ref{eq:Weyl:alpha}).
One can see from (\ref{eq:a-to-b}), see also \cite{NSS1} for more details, that the Weyl operator can be written in terms of the creation and annihilation operators $b_j^*$ and $b_j$, defined by the Bogoliubov transformation (\ref{eq:a-to-b}), as follows
\begin{equation}\label{eq:Weyl:VB}
\W(f)=\exp\left(\frac{i}{\sqrt{2}}(b(Vf)+b^*(Vf))\right), \text{ where }b(f)=\sum_{j=1}^{|\Lambda|}\overline{f_j}b_j
\end{equation}
and $V$ is given in (\ref{eq:Weyl:alpha}). It is well known that each eigenvector $\psi_\alpha$, with $\alpha\in\ell^\infty(\mathbb{N}_0^{|\Lambda|})$ is in the analytic vector for $b(f)+b^*(f)$, see e.g., Section 5.2.1.2 from \cite{BRvol2}. Thus one can define $\mathcal{W}(V^{-1}f)\psi_\alpha$ by power series expansion, see e.g., Theorem 8.30 in \cite{Weidmann},
\begin{equation}
\mathcal{W}(V^{-1}f)\psi_\alpha=\sum_{m\geq 0}\frac{(2^{-1/2}i)^m}{m!}(b(f)+b^*(f))^m\psi_\alpha.
\end{equation}
The Baker-Campbell-Hausdorff formula gives, using $[b(f),b^*(f)]=\|f\|^2$ and the CCR,
\begin{equation}
\langle\psi_\alpha|\W(V^{-1}f)|\psi_\alpha\rangle=e^{-\frac{1}{4}\|f\|^2} \left\langle\prod_{j=1}^{|\Lambda|}e^{\frac{-i}{\sqrt{2}}\overline{f_j}b_j}\psi_\alpha,
\prod_{j=1}^{|\Lambda|}
e^{\frac{i}{\sqrt{2}}\overline{f_j}b_j}\psi_\alpha\right\rangle
\end{equation}
where the exponential operators are defined by their power series expansions.
A direct calculation using the CCR of the $b_j$ operators gives
\begin{equation}
(b_j)^k\ \psi_\alpha=\left\{
                     \begin{array}{ll}
                       0, & \hbox{if $k>\alpha_j$;} \\
                       \sqrt{\frac{\alpha_j!}{(\alpha_j-k)!}}\ \psi_{\alpha-k e_j^T}, & \hbox{if $k\leq \alpha_j$}
                     \end{array}
                   \right.
\end{equation}
where $\{e_j\}_j$ are the canonical basis of $\ell^2(\mathbb{N}^{|\Lambda|})$. Meaning that,
\begin{equation}\label{eq:expB-taylor}
e^{\frac{i}{\sqrt{2}}\overline{f_j}b_j}\psi_\alpha=\sum_{k=0}^{\alpha_j}\frac{(2^{-1/2}i)^k\overline{f_j}^{k}}{k!}
\sqrt{\frac{\alpha_j!}{(\alpha_j-k)!}}\psi_{\alpha-ke^T_j}.
\end{equation}
By taking the expectations of $\mathcal{W}(V^{-1}f)$ at any eigenstate $\psi_\alpha$ using the power series expansions (\ref{eq:expB-taylor}) for $j=1,\ldots,|\Lambda|$,
\begin{eqnarray}\label{eq:Weyl:alpha:calc1}
\langle \W(V^{-1}f)\rangle_{\rho_\alpha}&=&
e^{-\frac{1}{4}\|f\|^2}\sum_{k_1,\tilde{k}_1=0}^{\alpha_1}\ldots
\sum_{k_{|\Lambda|},\tilde{k}_{|\Lambda|}=0}^{\alpha_{|\Lambda|}}
\prod_{j=1}^{|\Lambda|}\left(\frac{(2^{-1/2}i)^{k_j+\tilde{k}_j}f_j^{k_j}
\overline{f_j}^{\tilde{k}_j}}{k_j!\tilde{k}_j!}\right.
\times \nonumber\\
&&
\hspace{1.5cm}\left.\times\frac{\alpha_j!}{\sqrt{(\alpha_j-k_j)!(\alpha_j-\tilde{k}_j)!}}\right)\left\langle
\psi_{(\alpha_1-k_1,\ldots,\alpha_{|\Lambda|}-k_{|\Lambda|})},
\psi_{(\alpha_1-\tilde{k}_1,\ldots,\alpha_{|\Lambda|}-\tilde{k}_{|\Lambda|})}\right\rangle \nonumber\\
&=&e^{-\frac{1}{4}\|f\|^2}\prod_{j=1}^{|\Lambda|}\sum_{k_j=0}^{\alpha_j}
\frac{\left(-\frac{1}{2}|f_j|^2\right)^{k_j}}{k_j!}\binom{\alpha_j}{k_j}
\end{eqnarray}
where we used in the last step that $\{\psi_\alpha\}_{\alpha\in\mathbb{N}_0^{|\Lambda|}}$ are orthonormal and thus
\begin{equation}
\left\langle
\psi_{(\alpha_1-k_1,\ldots,\alpha_{|\Lambda|}-k_{|\Lambda|})},
\psi_{(\alpha_1-\tilde{k}_1,\ldots,\alpha_{|\Lambda|}-\tilde{k}_{|\Lambda|})}\right\rangle=
\prod_{j=1}^{|\Lambda|}\delta_{k_j,\tilde{k}_j},
\end{equation}
which completes the proof.
\end{proof}
Let us consider the orthonormal basis \begin{equation}\label{eq:def:ONB1}
\{|n\rangle;\ n\in\mathbb{N}_0\} \text{ of } \mathcal{L}^2(\R),
\end{equation}
 consisting of the eigenfunctions of the standard harmonic oscillator Hamiltonian $a^*a+\frac{1}{2}\idty$
written in terms of the creation and annihilation operators
\begin{equation}
a^*|n\rangle=\sqrt{n+1}|n+1\rangle, \text{ for } n\geq 0, \text{ and  }a|n\rangle=\sqrt{n}|n-1\rangle, \text{ for }n\geq 1.
\end{equation}
It is easy to see from Theorem \ref{thm:Weyl-Eigen} that the expectation of Weyl operator defined on $\mathcal{L}^2(\mathbb{R})$ generated by $z\in\mathbb{C}$ at the eigenstate $|n\rangle \langle n|$ associated with $n$ particles, is given by
\begin{equation}\label{eq:weyl:rho-n}
\langle \W_z\rangle_{n}:=\langle n|\W_z|n\rangle= L_n\left(\frac{|z|^2}{2}\right)e^{-\frac{1}{4}|z|^2},
\end{equation}
which is a simplified version of the well known result
\begin{equation}
\langle\W_z\rangle_{n}=e^{\frac{|z|^2}{4}}\sum_{m\geq 0}\frac{(-|z|^2)^m}{2^m(m!)^2}\frac{(n+m)!}{n!},
\end{equation}
 see e.g., XII.55 from \cite{Messiah}.

\subsection{Scaling the characteristic functions}
We will need later to deal with a scaled version of the characteristic functions of the eigenstates $\rho_n=|n\rangle\langle n|$ of the form $\langle\W_{\sqrt{a} z}\rangle_{n}$ for any positive integer $a$. The following theorem proves the existence of a unique operator $\rho_a^{(n)}$, such that $\langle\W_{\sqrt{a} z}\rangle_{n}=\langle\W_z\rangle_{\rho^{(n)}_a}$. In fact, the theorem defines this operator explicitly.
\begin{thm}\label{thm:rho1}
For any $a>0$ and $\ell\in\mathbb{N}_0$, there is a unique self-adjoint trace class operator $\rho^{(\ell)}_a$ on $\mathcal{L}^2(\R)$, for which
\begin{equation}\label{eq:Weyl-rhod-ell}
\langle \W_z\rangle_{\rho_a^{(\ell)}}=\langle\W_{\sqrt{a} z}\rangle_{\ell}=L_{\ell}\left(\frac{a|z|^2}{2}\right)e^{-\frac{1}{4}a|z|^2}, \quad \text{for all\ \ }z\in \C,
\end{equation}
where $L_\ell(\cdot)$ is the Laguerre polynomial of degree $\ell$. The operator $\rho_{a}^{(\ell)}$ is defined in terms of the orthonormal eigenvectors $\{|n\rangle, n\in\mathbb{N}_0\}$, given in (\ref{eq:def:ONB1}), as follows: for any  $n\in\mathbb{N}_0$
\begin{equation}\label{eq:rhod1}
 \rho_a^{(\ell)}\ |n\rangle=\sum_{j=0}^\ell \sigma_{j,\ell}(\zeta_a) \omega_{n,j,\ell}(\zeta_a)\ |n\rangle,\quad \zeta_a:=\frac{a-1}{a+1},
\end{equation}
where, for $x\in\mathbb{R}$,
$\sigma_{j,\ell}(x):=\binom{\ell}{j}(-x)^j(1+x)^{\ell-j}$, and
\begin{equation}\label{eq:def:sigma-omega}
\omega_{n,j,\ell}(x):=
\left\{
  \begin{array}{ll}
    \binom{n}{\ell-j}x^{n-(\ell-j)}(1-x)^{\ell-j+1} & \hbox{if }n\geq (\ell-j) \\
    0 & \hbox{otherwise.}
  \end{array}
\right.
\end{equation}
Moreover, we have the following bound for the trace norm of $\rho_a^{(\ell)}$
\begin{equation}\label{eq:def:f}
\|\rho_{a}^{(\ell)}\|_1\leq g_a(\ell):=\left\{
                                                              \begin{array}{ll}
                                                                a^\ell, & \hbox{if $a\geq 1$} \\
                                                                (1/a)^{\ell+1} & \hbox{if $a<1$.}
                                                              \end{array}
                                                            \right.
\end{equation}
\end{thm}
For every $a>0$, the number $\zeta_a$ satisfies $-1<\zeta_a<1$; and let's note that the function $g_a$ defined in (\ref{eq:def:f}) is an increasing function. For fixed $\ell$, $\zeta_a$, and $n$, $\{\sigma_{j,\ell}(\zeta_a),j=1,\ldots,\ell\}$ represent the weights of $\omega_{n,j,\ell}(\zeta_a)$ where $\sum_j\sigma_{j,\ell}(\zeta_a)=1$. Moreover, the identity
\begin{equation}\label{eq:eta-n}
\sum_{n=k}^\infty\binom{n}{k}x^{n-k}=
\frac{1}{(1-x)^{k+1}}, \quad -1<x<1,
\end{equation}
gives that $\sum_{n\geq 0}\omega_{n,j,\ell}(\zeta_a)=1$, meaning that $\Tr\rho_a^{(\ell)}=1$. We stress here that $\rho_a^{(\ell)}$ is not necessarily a state as some of its eigenvalues are negatives for certain values of $a$, $\zeta_a$, and  $\ell$. The cases $\ell\in\{0,1\}$ are presented after the proof.
\begin{proof} (of Theorem \ref{thm:rho1})
By construction, $\rho_a^{(\ell)}$ is self-adjoint and it is clear that it is trace class. To show the inequality (\ref{eq:def:f}), we start from (\ref{eq:rhod1}), then a calculation using the identity (\ref{eq:eta-n})
gives
\begin{equation}\label{eq:Tr1-rhod-ell-abs}
\|\rho_a^{(\ell)}\|_1\leq\sum_{j=0}^{\ell}
\binom{\ell}{j}|\zeta_a|^j(1+\zeta_a)^{\ell-j}
\left(\frac{1-\zeta_a}{1-|\zeta_a|}\right)^{\ell-j+1}.
\end{equation}
This yields
\begin{equation}
\|\rho_a^{(\ell)}\|_1\leq\left\{
                           \begin{array}{ll}
                             (1+2\zeta_a)^\ell, & \hbox{if } \zeta_a\geq 0\\
                             \frac{1-\zeta_a}{1+\zeta_a}\ (1-2\zeta_a)^\ell, & \hbox{if } \zeta_a<0.
                           \end{array}
                         \right.
\end{equation}
Using $\zeta_a=\frac{a-1}{a+1}$ and noting that $\zeta_a\geq 0$ if and only if $a\geq 1$, we get the desired bound (\ref{eq:def:f}).

Next, we prove (\ref{eq:Weyl-rhod-ell}). We expand the trace in $\langle \W_z\rangle_{\rho^{(\ell)}_a}$  over the orthonormal basis $\{|n\rangle,\ n\in\mathbb{N}_0\}$ given in (\ref{eq:def:ONB1})
to get, using (\ref{eq:rhod1}) and (\ref{eq:def:sigma-omega}),
\begin{equation}\label{eq:Weyl:rhodl}
\langle \W_z\rangle_{\rho^{(\ell)}_a}=\sum_{n=0}^{\infty}\langle n|\W_z\rho^{(\ell)}_a|n\rangle
=\sum_{j=0}^{\ell}\sigma_{j,\ell}(\zeta_a)
\sum_{n=\ell-j}^\infty \omega_{n,j,\ell}(\zeta_a)\langle n|\W_z|n\rangle.
\end{equation}
The \textit{first  multiplication theorem} of Erd\'{e}lyi,  see e.g., \cite{Truesdell}, says that for any $x\in\R$ and $k\geq 0$,
\begin{equation}
\sum_{n=k}^\infty \binom{n}{k}\zeta_a^{n-k}L_n(x)=
\frac{e^{-\frac{\zeta_a}{1-\zeta_a}x}}{(1-\zeta_a)^{k+1}}
L_k\left(\frac{x}{1-\zeta_a}\right),
\end{equation}
which gives directly that
\begin{equation}\label{eq:weyl:Sum:wWeyl}
\sum_{n=\ell-j}^\infty \omega_{n,j,\ell}(\zeta_a)\langle n|\W_z|n\rangle=e^{-\frac{1}{4}\left(\frac{1+\zeta_a}{1-\zeta_a}\right)|z|^2}
L_{\ell-j}\left(\frac{1}{1-\zeta_a}\frac{|z|^2}{2}\right).
\end{equation}
Moreover, one can see that for all $x\in\R$
\begin{equation}\label{eq:weyl:sigma}
\sum_{j=0}^\ell\sigma_{j,\ell}(\zeta_a)L_{\ell-j}(x)=L_\ell((1+\zeta_a)x),
\end{equation}
which follows by expanding the $(\ell-j)$-th degree Laguerre polynomial $L_{\ell-j}(x)$, as in (\ref{eq:def:Laguerre}), then the change of order of the two summations yields directly the right hand side.

Substitute (\ref{eq:weyl:Sum:wWeyl}) in (\ref{eq:Weyl:rhodl}) then apply (\ref{eq:weyl:sigma}) with $x=|z|^2/2$,
\begin{equation}
\langle \W_z\rangle_{\rho^{(\ell)}_a}=
e^{-\frac{1}{4}\left(\frac{1+\zeta_a}{1-\zeta_a}\right)|z|^2}
L_{\ell}\left(\frac{1+\zeta_a}{1-\zeta_a}\frac{|z|^2}{2}\right)=
e^{-\frac{1}{4}a|z|^2}
L_{\ell}\left(\frac{a|z|^2}{2}\right)
\end{equation}
The uniqueness follows from Lemma 3.1 in \cite{NSS2}.
\end{proof}
When $a=1$, the operator $\rho_1^{(\ell)}$, defined in (\ref{eq:rhod1}), is the eigenstate $|\ell\rangle\langle\ell|$ of the standard harmonic oscillator as in (\ref{eq:weyl:rho-n}).  Moreover, one can see, from (\ref{eq:Tr1-rhod-ell-abs}), that the equality in (\ref{eq:def:f}) is attained if and only if $\ell=0$. This makes Lemma 3.5 in \cite{NSS2} a special case of Theorem \ref{thm:rho1}, where
the operator $\rho_a^{(0)}$ is defined as
\begin{equation}\label{eq:rhod:0}
\rho_a^{(0)}|n\rangle = \lambda_{n,a}^{(0)}\, |n\rangle, \quad\text{where\, }\lambda_{n,a}^{(0)}:=(1-\zeta_a)(\zeta_a)^n, \, n\in\mathbb{N}_0.
\end{equation}
Seeing that the eigenvalues in (\ref{eq:rhod1}) look somehow  cumbersome,  we present the first nontrivial case $\ell=1$, (\ref{eq:rhod1}) reduces to
$\rho_a^{(1)}|n\rangle = \lambda_{n,a}^{(1)}\, |n\rangle$, where
\begin{equation}\label{eq:sgn:1}
\lambda_{n,a}^{(1)}:=
\left\{
  \begin{array}{ll}
    -\zeta_a(1-\zeta_a) & n=0 \\
    n(1+\zeta_a)(1-\zeta_a)^2\zeta_a^{n-1}-(1-\zeta_a)\zeta_a^{n+1} & n=1,2,\ldots
  \end{array}
\right.
\end{equation}
and it is straight forward to check that for even $n$,
\begin{equation}
\sgn(\lambda_{n,a}^{(1)})=\left\{
                            \begin{array}{ll}
                              \sgn(\zeta_a) & n>\frac{\zeta_a^2}{1-\zeta_a^2} \\
                              -\sgn(\zeta_a) & 0\leq n<\frac{\zeta_a^2}{1-\zeta_a^2}
                            \end{array}
                          \right.
\end{equation}
where $\sgn(\cdot)$ is the sign function.
\section{Diagonalization of the partial transpose}\label{sec:PT:diag}
In this section we diagonalize  the partial transpose of $\rho_N$ given in (\ref{eq:def:rhoN}). This is summarized in the following Theorem.
\begin{thm}\label{thm:rhoT1:eigen}
Fix $\Lambda_0\subset\Lambda$. There exists a unitary $U\in\mathcal{B}(\mathcal{H}_\Lambda)$ such that, for any $N\in\mathbb{N}_0$, and the corresponding set $\J_N$ given in (\ref{eq:def:Jn}), the partial transpose of the $N$-modes ensemble $\rho_N$ with respect to the decomposition (\ref{eq:decomposition:H}) is given as
\begin{equation}\label{eq:rhoT1-diag}
\rho_N^{T_1}=\frac{1}{|\J_N|}\sum_{\alpha\in \J_N}U\left(\bigotimes_{k=1}^{|\Lambda|}\rho^{(\alpha_k)}_{d_k}\right)U^*
\end{equation}
where for each $\alpha=(\alpha_1,\ldots,\alpha_{|\Lambda|})\in\J_N$, the operators $\{\rho_{d_k}^{(\alpha_k)}, k=1,\ldots,|\Lambda|\}$ are as in Theorem \ref{thm:rho1}, with $\{d_k>0, k=1,\ldots,|\Lambda|\}$ being the symplectic eigenvalues of
\begin{equation}\label{eq:tildeM}
\tilde{M}:=\begin{pmatrix}
             \idty & 0 \\
             0 & \mathbb{P} \\
           \end{pmatrix}M\begin{pmatrix}
             \idty & 0 \\
             0 & \mathbb{P} \\
           \end{pmatrix}
\end{equation} and $\mathbb{P}$ is the diagonal operator with diagonal entries \begin{equation}\label{eq:P}
\mathbb{P}_{xx}=\left\{
                  \begin{array}{ll}
                    -1 & \hbox{if $x\in\Lambda_0$} \\
                    1 & \hbox{otherwise.}
                  \end{array}
                \right.
\end{equation}
\end{thm}
Note that for any $\alpha\in\J_N$ and all symplectic eigenvalues $\{d_k,k=1,\ldots,|\Lambda|\}$, the operators $\{\rho_{d_k}^{(\alpha_k)}, k=1,\ldots,|\Lambda|\}$ are simultaneously diagonalizable with known eigenvalues given in  (\ref{eq:rhod1}), meaning that (\ref{eq:rhoT1-diag})  is an exact diagonalization of the partial transpose of $\rho_N$. In particular, the eigenvalues of $\rho_N^{T_1}$ are $\{\lambda_{n_1,\ldots,n_{|\Lambda|}};\ n_1,\ldots,n_{|\Lambda|}\in\mathbb{N}_0\}$ where
\begin{equation}
\lambda_{n_1,\ldots,n_{|\Lambda|}}=\frac{1}{|\J_N|}\sum_{\alpha\in\J_N}
\prod_{k=1}^{|\Lambda|}\left(\sum_{j=0}^{\alpha_k}
\sigma_{j,\alpha_k}(\zeta_{d_k})\omega_{n_k,j,\alpha_k}(\zeta_{d_k})\right).
\end{equation}
Since $\Tr\rho_{a}^{(\ell)}=1$ for any $a>0$ and $\ell\in\mathbb{N}_0$, we find that $\Tr\rho_N^{T_1}=1$, as it should be. Since it is clear and guaranteed that for all $N\in\mathbb{N}_0$, the state $\rho_N$ is entangled. And due to the complexity of $\lambda_{n_1,\ldots,n_{|\Lambda|}}$, one may want to check whether the Peres-Horodecki criterion is  satisfied \cite{PeresPPT1,HorodeckiPPT2}. Its says that the state is entangled whenever its partial transpose has some negative eigenvalues, the opposite is not generally correct, meaning that some entangled states may not be detected by the logarithmic negativity.
The following example deals with the case $N=1$, and describes the negative eigenvalues of $\rho_{N=1}^{T_1}$.
\begin{ex}
When $N=1$, $\J_{N=1}=\{e_j^T;\ j=1,\ldots,|\Lambda|\}$, here $\{e_j\}_j$ are the canonical basis of $\ell^2(\Lambda)$. The state $\rho_{N=1}$ is the statistical ensemble of the eigenstates of $H_\Lambda$ that are associated with exactly one excitation in the $j$-th vertex. Theorem \ref{thm:rho1} gives
\begin{equation}
\rho_{N=1}^{T_1}=\frac{1}{N}\sum_{j=1}^{|\Lambda|}
U\left(\rho_{d_j}^{(1)}\otimes \bigotimes_{k=1,\, k\neq j}^{|\Lambda|}\rho_{d_k}^{(0)}\right)U^*.
\end{equation}
Using (\ref{eq:rhod:0}) and (\ref{eq:sgn:1}), one can see that the eigenvalues $\lambda_{n_1,\ldots,n_{|\Lambda|}}$ of $\rho_{N=1}^{T_1}$ are negative, for example,  for the indices $(n_1,\ldots,n_{|\Lambda|})\in\mathbb{N}_0$ satisfying
\begin{equation}
n_k=\left\{
      \begin{array}{ll}
        0, & \hbox{if } d_k>1 \\
        \geq \lceil\frac{\zeta_{d_k}^2}{1-\zeta_{d_k}^2}\rceil+1 \text{ and even} & \hbox{if } d_k<1.
      \end{array}
    \right.
\end{equation}
\end{ex}
We will need to use the following lemma that provides a crucial part to the proof of Theorem \ref{thm:rhoT1:eigen}.
\begin{lem}\label{lem:I-N}
Fix $N\in\mathbb{N}_0$, and consider the corresponding set $\J_N$ defined in (\ref{eq:def:Jn}) then for any $x_1,\ldots,x_{|\Lambda|}\in\mathbb{R}$,
\begin{equation}\label{eq:I-N}
\sum_{\alpha\in \J_N}\prod_{k=1}^{|\Lambda|}L_{\alpha_k}(x_k)=\mathcal{Q}_N\left(\sum_{j=1}^{|\Lambda|}x_k\right),
\end{equation}
where $L_n(\cdot)$ is the Laguerre polynomial of degree $n$ and
\[
\mathcal{Q}_N(x):=L^{(|\Lambda|-1)}_N(x)=\sum_{j=0}^N (-1)^j\binom{N+|\Lambda|-1}{N-j}\frac{x^j}{j!},\ \ x\in\mathbb{R}
\]
is the $(|\Lambda|-1)$-generalized Laguerre polynomial of degree $N$.
\end{lem}
We include the proof of this lemma in Appendix \ref{proof:lemmaLaguerre-I-N}.
\begin{proof}(of Theorem \ref{thm:rhoT1:eigen}) A direct calculation gives that the expectation of the Weyl operators at the partial transpose of any state $\rho$ is given as
\begin{equation}\label{eq:Weyl:PT}
\langle \W(f)\rangle_{\rho^{T_1}}=\left\langle \W\left(\begin{pmatrix}
                                                    \idty & 0 \\
                                                    0 & \mathbb{P} \\
                                                  \end{pmatrix}
\tilde{f}\right)\right\rangle_{\rho}
\end{equation}
where $\mathbb{P}$ is the diagonal operator with diagonal elements given in (\ref{eq:P}). Let $\rho_\alpha$ be any eigenstate of $H_\Lambda$ associated with occupation number vector $\alpha\in\mathbb{N}_0^{|\Lambda|}$. Theorem \ref{thm:Weyl-Eigen} and (\ref{eq:Weyl:PT}) instantly give
\begin{equation}\label{eq:Weyl:eig:T1}
\langle \W(f)\rangle_{\rho_{\alpha}^{T_1}}=e^{-\frac{1}{4}\langle\tilde{f},\tilde{M}\tilde{f}\rangle}\prod_{k=1}^{| \Lambda|} L_{\alpha_k} \left(\frac{\langle\tilde{f},\tilde{M}\chi_k(\tilde{M})\tilde{f}\rangle}{2} \right) ,
\end{equation}
where $\tilde{M}$ is as in (\ref{eq:tildeM}) and $\chi_k(\tilde{M})$ is the spectral projection of $\tilde{M}$ defined the same way as $\chi_k(M)$ in (\ref{eq:chikM}).

Next, for the $N$-modes ensemble $\rho_N$, given in (\ref{eq:def:rhoN}), and for any $f\in\ell^2(\Lambda)$, the linearity of the trace and the partial transpose,  equation (\ref{eq:Weyl:eig:T1}), and Lemma \ref{lem:I-N} give
\begin{equation}\label{eq:Weyl:rhoN-M}
\langle \W(f)\rangle_{\rho_N^{T_1}}= \frac{1}{|\J_N|}\mathcal{Q}_N\left(\frac{\langle\tilde{f},\tilde{M}\tilde{f}\rangle}{2}\right)
e^{-\frac{1}{4}\langle\tilde{f},\tilde{M}\tilde{f}\rangle},
\end{equation}
where we used the resolution of the identity $\sum_{k=1}^{|\Lambda|}\chi_k(\tilde{M})=\idty$.
Since $\tilde{M}$ is real symmetric and positive definite (almost surely) then by the Williamson Theorem, see e.g., Theorem 8.11 in \cite{deGosson}, there exists a $2|\Lambda|\times 2|\Lambda|$ symplectic $S$ such that
\begin{equation}\label{eq:SMS}
S^T\tilde{M}S=\begin{pmatrix}
                D & 0 \\
                0 & D \\
              \end{pmatrix},\ \text{ where } D:=\diag\{d_1,\ldots,d_{|\Lambda|}\}.
\end{equation}
Here $d_j>0$ for all $j$, are the symplectic eigenvalues of $\tilde{M}$, which are the positive eigenvalues of the hermitian matrix $i\tilde{M}^{\frac{1}{2}}J\tilde{M}^{\frac{1}{2}}$, where $J=\begin{pmatrix}
           0 & -\idty \\
           \idty & 0 \\
         \end{pmatrix}
$. Note that the matrix $S$ does not depend on $N$ or $\J_N$.

Symplectic $S$ induces a unitary $U\in\mathcal{B}(\mathcal{H}_\Lambda)$ such that, see e.g., \cite{BD07,Shale62}
\begin{equation}\label{eq:SU}
\W(S\tilde{f})=U \W(\tilde{f}) U^*.
\end{equation}
This means that the expectation of the Weyl operator at $U^*\rho_N^{T_1}U$ is given as
\begin{equation}
\langle \W(f)\rangle_{U^*\rho_N^{T_1}U}=\langle \W(S\tilde{f})\rangle_{\rho_N^{T_1}}.
\end{equation}
Equation (\ref{eq:Weyl:rhoN-M}) and the symplectic diagonalization (\ref{eq:SMS}) give
\begin{eqnarray}\label{eq:Weyl-diag}
\langle \W(f)\rangle_{U^*\rho_N^{T_1}U}&=& \frac{1}{|\J_N|}\mathcal{Q}_N\left(\frac{1}{2}\langle\tilde{f},\begin{pmatrix}D & 0 \\ 0 & D \\\end{pmatrix}\tilde{f}\rangle\right)
\exp\left(-\frac{1}{4}\langle\tilde{f},\begin{pmatrix}D & 0 \\ 0 & D \\\end{pmatrix}\tilde{f}\rangle\right)\\
&=&
\frac{1}{|\J_N|}\mathcal{Q}_N\left(\sum_{j=1}^{|\Lambda|}\frac{d_k|f_k|^2}{2}\right)
e^{-\frac{1}{4}\sum_{j=1}^{|\Lambda|}d_k|f_k|^2}\nonumber.
\end{eqnarray}
Lemma \ref{lem:I-N} allows to write the polynomial $\mathcal{Q}_N(\cdot)$ back as a sum of products of Laguerre polynomials,
\begin{equation}
\langle \W(f)\rangle_{U^*\rho_N^{T_1}U}= \frac{1}{|\J_N|}\sum_{\alpha\in \J_N}\prod_{k=1}^{|\Lambda|}\left(L_{\alpha_k}\left(\frac{d_k|f_k|^2}{2}\right)e^{-\frac{1}{4}d_k|f_k|^2}\right).
\end{equation}
Theorem \ref{thm:rho1} defines the operators $\{\rho_{d_k}^{(\alpha_k)}\}_k$ for which
\begin{equation}
\langle \W(f)\rangle_{U^*\rho_N^{T_1}U}= \frac{1}{|\J_N|}\sum_{\alpha\in \J_N}\prod_{k=1}^{|\Lambda|}\langle \W_{f_k}\rangle_{\rho^{(\alpha_k)}_{d_k}}
= \frac{1}{|\J_N|}\sum_{\alpha\in \J_N}\Tr\left[\W(f)\bigotimes_{k=1}^{|\Lambda|}\rho^{(\alpha_k)}_{d_k}\right]
\end{equation}
where we used (\ref{eq:Weyl:tensor}). The linearity of the trace instantly yields
\begin{equation}
\langle \W(f)\rangle_{U^*\rho_N^{T_1}U}= \Tr\left[\W(f)\left(\frac{1}{|\J_N|}\sum_{\alpha\in \J_N}\bigotimes_{k=1}^{|\Lambda|}\rho^{(\alpha_k)}_{d_k}\right)\right]
\end{equation}
and since this is true for any $f\in\ell^2(\Lambda)$ then (\ref{eq:rhoT1-diag}) follows directly by Lemma 3.1 from \cite{NSS2}.
\end{proof}
We discuss here the difficulty we ran into when attempting the diagonalization of the energy eigenstates $\rho_\alpha$. The matrices in the family $\{\tilde{M}\chi_k(\tilde{M}),k=1,\ldots,|\Lambda|\}$ from (\ref{eq:Weyl:eig:T1}) are symmetric non-negative, but not strictly positive. Meaning that none of them is symplectic diagonalizable in the sense of Williamson Theorem as in (\ref{eq:SMS}). Our attempts to get around this problem did not succeed. The problem of obtaining bounds for the entanglements of the energy eigenstates of $H_\Lambda$ still stands as an interesting open problem.

\section{Proving the upper bound for the logarithmic negativity}\label{sec:LN:bound}
In this section we prove Theorem \ref{thm:AreaLaw} and Corollary \ref{cor:AreaLaw2}. We start with the following lemma that provides a bound for the logarithmic negativity of $\rho_N$ in terms of the absolute values of the matrix elements of $h_\Lambda^{-\frac{1}{2}}$.
\begin{lem}\label{lem:LN:upper1}
Fix $\Lambda_0\subset\Lambda$, for any $N\in\mathbb{N}_0$, and the corresponding N-modes ensemble $\rho_N$ as in (\ref{eq:def:rhoN}), we have the following bound for the logarithmic negativity of $\rho_N$ with respect to the decomposition (\ref{eq:decomposition:H}),
\begin{equation}\label{eq:LN:upper}
\mathcal{N}(\rho_N)\leq (2N+1)\|h_{\Lambda}^{\frac{1}{2}}\|\sum_{x\in\Lambda_0}
\sum_{y\in\Lambda\setminus\Lambda_0}|\langle\delta_x,h_{\Lambda}^{-\frac{1}{2}}\delta_y\rangle|.
\end{equation}
\end{lem}
Again, in the degenerate case $N=0$, (\ref{eq:LN:upper}) reduces to the well known bound of the logarithmic negativity of the ground state, e.g., \cite{NSS2}.
\begin{proof}
The diagonalization of the partial transpose of $\rho_N$ in (\ref{eq:rhoT1-diag}) gives directly that
\begin{equation}
\|\rho^{T_1}_N\|_1\leq \frac{1}{|\J_N|}\sum_{\alpha\in \J_N} \prod_{k=1}^{|\Lambda|}\|\rho_{d_k}^{(\alpha_k)}\|_1.
\end{equation}
We recall here that the operators $\{\rho_{d_k}^{(\alpha_k)},\, k=1,\ldots,|\Lambda|\}$ are as in Theorem \ref{thm:rho1}, where $\{d_k>0,\, k=1,\ldots,|\Lambda|\}$ are the symplectic eigenvalues of $\tilde{M}$ given in (\ref{eq:tildeM}) and  $\alpha=(\alpha_1,\ldots,\alpha_{|\Lambda|})$ is the occupation number vector.
Theorem \ref{thm:rho1} implies that for every $k=1,\ldots,|\Lambda|$, and $\alpha=(\alpha_1,\ldots,\alpha_{|\Lambda|})\in\J_N$, we have the following bound
\begin{equation}
\max_{\alpha\in\J_N}\|\rho_{d_k}^{(\alpha_k)}\|_1\leq \max_{\alpha\in\J_N}g_{d_k}(\|\alpha\|_\infty)\leq g_{d_k}(N),
\end{equation}
where the functions $\{g_{d_k}\}_k$ are defined in (\ref{eq:def:f}),
meaning that
\begin{equation}\label{eq:TN:T1:upperbound}
\|\rho^{T_1}_N\|_1\leq \prod_{k=1}^{|\Lambda|} g_{d_k}(N).
\end{equation}
Thus, the logarithmic negativity can be bounded by
\begin{equation}\label{eq:LN:sums1}
\mathcal{N}(\rho_N)=\log\|\rho^{T_1}\|_1\leq \sum_{k=1}^{|\Lambda|}\log g_{d_k}(N)=\sum_{k: d_k\geq 1}\log d_k^{N}+\sum_{k: d_k<1}\log d_k^{-(N+1)}.
\end{equation}
Recall that $\{d_k\}_k$ are the positive eigenvalues of the hermitian
$L:=i\tilde{M}^{\frac{1}{2}}J\tilde{M}^{\frac{1}{2}}$
that has symmetric spectrum about zero (using $L^T=L$). A direct calculation gives
\begin{equation}
L^2=\begin{pmatrix}
      Z^{-1} & 0 \\
      0 & \mathbb{P}Z\mathbb{P} \\
    \end{pmatrix},
\end{equation}
where $Z=h_\Lambda^{\frac{1}{4}}\mathbb{P}h_\Lambda^{-\frac{1}{2}}
\mathbb{P}h_\Lambda^{\frac{1}{4}}$. Since $Z^{-1}$ is similar to its inverse $Z$, and hence to $\mathbb{P}Z\mathbb{P}$, the symplectic eigenvalues of $\tilde{M}$ are the square roots of the eigenvalues of one the diagonal block of $L^2$. For the sum over $\{k:d_k\geq 1\}$ in (\ref{eq:LN:sums1}), we use the matrix $\mathbb{P}Z\mathbb{P}$, and we use the matrix $Z^{-1}$ in the second sum. This gives the bound
\begin{equation}\label{eq:LN:upper1}
\mathcal{N}(\rho_N)\leq\frac{1}{2}\Tr\left[\mathcal{P}_N
\log\left(Z\right)\right]\text{ \,where\, }\mathcal{P}_N:=N\mathbb{P}P^{>1}\mathbb{P}+(N+1)P^{>1}
\end{equation}
and  $P^{>1}$ is the orthogonal projections onto the subspace $Z>1$.
Next, we proceed as in \cite{NSS2} and we include it for the sake of completeness. We rewrite $Z$ as follows
\begin{equation}
Z=h_\Lambda^{\frac{1}{4}}\mathbb{P}h_\Lambda^{-\frac{1}{2}}
\mathbb{P}h_\Lambda^{\frac{1}{4}}=
h_\Lambda^{\frac{1}{4}}\mathbb{P}\left[h_\Lambda^{-\frac{1}{2}},
\mathbb{P}\right]h_\Lambda^{\frac{1}{4}}+\idty,
\end{equation}
and we use the concavity of the logarithm to get
\begin{equation}
2\mathcal{N}(\rho_N)\leq \Tr\left[\mathcal{P}_N\left(
h_\Lambda^{\frac{1}{4}}\mathbb{P}[h_\Lambda^{-\frac{1}{2}},
\mathbb{P}]h_\Lambda^{\frac{1}{4}}
\right)\right]\leq \|h_\Lambda^{\frac{1}{4}}\mathcal{P}_Nh_\Lambda^{\frac{1}{4}}\|
\|\mathbb{P}[h_\Lambda^{-\frac{1}{2}},
\mathbb{P}]\|_1.
\end{equation}
We bound the $1$-norm by the sum of the absolute values of the matrix elements in any basis, and we note that
\begin{equation}
\|h_\Lambda^{\frac{1}{4}}\mathcal{P}_Nh_\Lambda^{\frac{1}{4}}\|\leq (2N+1)\|h_\Lambda^{\frac{1}{4}}\|^2=(2N+1)\|h_\Lambda^{\frac{1}{2}}\|
\end{equation}
to obtain
\begin{equation}
2\mathcal{N}(\rho_N)\leq (2N+1)\|h_\Lambda^{\frac{1}{2}}\|\sum_{x,y\in\Lambda}
|\langle\delta_x,\mathbb{P}[h_\Lambda^{-\frac{1}{2}},
\mathbb{P}]\delta_y\rangle|.
\end{equation}
The desired bound is then proven by observing that
\begin{equation}
\langle\delta_x,\mathbb{P}[h_\Lambda^{-\frac{1}{2}},\mathbb{P}]\delta_y\rangle
=\left\{
   \begin{array}{ll}
     0 & \hbox{if } x,y\in\Lambda_0\text{ or }x,y\in\Lambda\setminus\Lambda_0 \\
     -2\langle\delta_x,h_\Lambda^{-\frac{1}{2}}\delta_y\rangle & \hbox{otherwise.}
   \end{array}
 \right.
\end{equation}
\end{proof}
Next, we present the proof of the area law for the $N$-modes ensemble $\rho_N$.
\begin{proof} (of Theorem \ref{thm:AreaLaw})
By averaging the disorder in the bound in Lemma \ref{lem:LN:upper1} and using (\ref{eq:def:spectrumh}) and the eigencorrelator localization (\ref{eq:ecl}) we obtain
\begin{equation}
\mathbb{E}(\mathcal{N}(\rho_N))\leq C(2N+1)(4d\lambda+k_{\max})^{1/2}\sum_{x\in\Lambda_0,\ y\in\Lambda\setminus\Lambda_0}e^{-\mu|x-y|}.
\end{equation}
For each $x\in\Lambda_0$ and $y\in\Lambda\setminus\Lambda_0$ there is at least one $z\in\partial\Lambda_0$ such that $|x-y|=|x-z|+|y-z|$, thus
\begin{eqnarray}
\mathbb{E}(\mathcal{N}(\rho_N))&\leq& C(2N+1)(4d\lambda+k_{\max})^{1/2}\sum_{z\in\partial\Lambda_0}\sum_{{\tiny
\begin{array}{c}
  x\in\Lambda_0,\ y\in\Lambda\setminus\Lambda_0 \\
  |x-y|=|x-z|+|y-z|
\end{array}}}e^{-\mu|x-z|}e^{-\mu|y-z|}\\
&\leq& C(2N+1)(4d\lambda+k_{\max})^{1/2}\sum_{z\in\partial\Lambda_0}
\left(\sum_{x\in\mathbb{Z}^d}e^{-\mu|x-z|}\right)^2\nonumber,
\end{eqnarray}
which gives the bound in (\ref{eq:AreaLaw}) with an explicit value for the constant $\tilde{C}$.
\end{proof}
We finally comment on how the above proof of Lemma \ref{lem:LN:upper1} changes to prove Corollary \ref{cor:AreaLaw2}.
\begin{proof} (of Corollary \ref{cor:AreaLaw2})
First, note that the symplectic matrix $S$, that diagonalizes $\tilde{M}$ in (\ref{eq:SMS}), does not depend on $N$, giving that its induced unitary operator $U$ in (\ref{eq:SU}) is also independent on $N$. Thus, (\ref{eq:rhoT1-diag}) gives that for any $L\in\mathbb{N}_0$,
\begin{equation}
U^*\rho_L^{T_1}U=\frac{1}{|\J_L|}\sum_{\alpha\in \J_L}\bigotimes_{k=1}^{|\Lambda|}\rho_{d_k}^{(\alpha_k)}.
\end{equation}
Thus, for any $\omega\in\Omega_N$, defined in (\ref{eq:def:Omega}), and the corresponding state $\rho_{\omega,N}$, defined in (\ref{eq:def:tilderhoN}), and using the linearity of the partial transpose,
\begin{equation}
U^*\rho_{\omega,N}^{T_1}U=\sum_{L=0}^{N}
\frac{\omega_L}{|\J_L|}
\sum_{\alpha^{(L)}\in |\J_L|}\bigotimes_{k=1}^{|\Lambda|}
\rho_{d_k}^{(\alpha^{(L)}(k))}.
\end{equation}
Theorem  \ref{thm:rho1} gives that
$
\|\rho_{d_k}^{(\alpha^{(L)}(k))}\|_1\leq g_{d_k}(\|\alpha^{(L)}\|_\infty)\leq g_{d_k}(N)
$
for all  $k=1,\ldots,|\Lambda|$,
where the functions $\{g_{d_k}\}_k$ are as in (\ref{eq:def:f}).
Meaning that the trace norm of $\rho_{\omega,N}^{T_1}$ is bounded as
$
\|\rho_{\omega,N}^{T_1}\|_1\leq
\prod_{k=1}^{|\Lambda|}g_{d_k}(N)
$
which coincides with the upper bound for the trace norm of the partial transpose of $\rho_N$ in (\ref{eq:TN:T1:upperbound}).
\end{proof}
\appendix
\section{The expected energy at the $N$-modes ensemble}\label{sec:App:A}
In this appendix we prove the following lemma,
\begin{lem}\label{sec:App:Energy}
For any $N\in\mathbb{N}_0$ and the corresponding $N$-modes ensemble $\rho_N$ defined in (\ref{eq:def:rhoN}), the expected energy of the harmonic Hamiltonian $H_\Lambda$ given in (\ref{eq:def:H}) at $\rho_N$ is
\begin{equation}
\langle H_\Lambda\rangle_{\rho_N}=\sum_{k=1}^{|\Lambda|}\gamma_k\left(1+\frac{2N}{|\Lambda|}\right),
\end{equation}
where $\{\gamma_k\}_k$ are the eigenvalues of $h_\Lambda^{\frac{1}{2}}$ defined in (\ref{eq:def:h}).
\end{lem}
\begin{proof}
Using (\ref{eq:def:rhoN}) and (\ref{eq:H:eigenvalues}) we get
\begin{equation}
\langle H_\Lambda\rangle_{\rho_N}=\frac{1}{|\J_N|}\sum_{\alpha\in\J_N}\sum_{k=1}^{|\Lambda|}\gamma_k(2\alpha_k+1)
=\sum_{k=1}^{|\Lambda|}\gamma_k+\frac{2}{|\J_N|}\sum_{\alpha\in\J_N}\sum_{k=1}^{|\Lambda|}\gamma_k\alpha_k.
\end{equation}
And due to the symmetry of the set $\J_N$, we get
\begin{equation}
\langle H_\Lambda\rangle_{\rho_N}=\sum_{k=1}^{|\Lambda|}\gamma_k\left(1+\frac{2}{|\J_N|}\sum_{\alpha\in\J_N} \alpha_1\right).
\end{equation}
A counting argument gives
\begin{equation}\label{eq:energy:pf:main-frml}
|\J_N|=\binom{N+|\Lambda|-1}{N}\quad\text{and } \sum_{\alpha\in\J_N} \alpha_1=\binom{N+|\Lambda|-1}{|\Lambda|},
\end{equation}
which gives directly the desired formula.
The formula for $|\J_N|$ follows from the ``multiset coefficient'' or the ``stars and bars'' in the combinatorial mathematics \cite{Feller50,Stanley97}.

To find the second sum in (\ref{eq:energy:pf:main-frml}), $N$ particles need to be decomposed as $(N-k)+k$ with $0\leq k\leq N-1$, where we need to place $(N-k)$ particles in the first slot and $k$ particles in the rest $(|\Lambda|-1)$ slots. The latter is performed, using the first result in (\ref{eq:energy:pf:main-frml}), in $\binom{k+|\Lambda|-2}{k}$ ways, thus
\begin{equation}
\sum_{\alpha\in\J_N}\alpha_1=\sum_{k=0}^{N-1}(N-k)
\binom{k+|\Lambda|-2}{k}.
\end{equation}
A careful inspection of this sum using the elementary formula
\begin{equation}
\sum_{\alpha=0}^{n}\binom{\alpha+\ell}{\alpha}=\binom{n+\ell+1}{n}
\end{equation}
 gives the desired formula.
\end{proof}
\section{Proof of Lemma \ref{lem:I-N}}\label{proof:lemmaLaguerre-I-N}
Using (\ref{eq:def:Jn}) and (\ref{eq:def:Laguerre}), the left hand side of (\ref{eq:I-N}) expands as
\begin{equation}
\sum_{\alpha\in\J_N}\prod_{k=1}^{|\Lambda|}L_{\alpha_k}(x_k)=\sum_{\alpha_1+\ldots+\alpha_{|\Lambda|}=N}\ \ \sum_{\ell_1,\ldots,\ell_{|\Lambda|}=0}^{\alpha_1,\ldots,\alpha_{|\Lambda|}}\ \
(-1)^{\sum_j\ell_j}\prod_{j=1}^{|\Lambda|}\binom{\alpha_j}
{\ell_j}\frac{x_1^{\ell_1}\ldots x_{|\Lambda|}^{\ell_{|\Lambda|}}}{\ell_1!\ldots \ell_{|\Lambda|}!}.
\end{equation}
The terms in this sum of degree $0\leq K\leq N$ are precisely given as
\begin{equation}
(-1)^{K}\sum_{\alpha_1+\ldots+\alpha_{|\Lambda|}=N}\ \
\sum_{{\tiny\begin{array}{c}
        \ell_1,\ldots,\ell_{|\Lambda|}=0 \\
        \ell_1+\ldots+\ell_{|\Lambda|}=K
      \end{array}}
}^{\alpha_1,\ldots,\alpha_{|\Lambda|}}\ \
\prod_{j=1}^{|\Lambda|}\binom{\alpha_j}{\ell_j}
\frac{x_1^{\ell_1}\ldots x_{|\Lambda|}^{\ell_{|\Lambda|}}}{\ell_1!\ldots \ell_{|\Lambda|}!}.
\end{equation}
Then by changing the order of the two big summations
\begin{equation}
(-1)^{K}\sum_{{\tiny\begin{array}{c}
        \ell_1,\ldots,\ell_{|\Lambda|}=0 \\
        \ell_1+\ldots+\ell_{|\Lambda|}=K
      \end{array}}
}^{N}\ \ \frac{x_1^{\ell_1}\ldots x_{|\Lambda|}^{\ell_{|\Lambda|}}}{\ell_1!\ldots \ell_{|\Lambda|}!}\ \
\sum_{{\tiny\begin{array}{c}
        \alpha_1=\ell_1,\ldots,\alpha_{|\Lambda|}=
        \ell_{|\Lambda|} \\
        \alpha_1+\ldots+\alpha_{|\Lambda|}=N\\
      \end{array}}}^N\ \
\prod_{j=1}^{|\Lambda|}\binom{\alpha_j}{\ell_j}.
\end{equation}
The big sum over the $\alpha$'s is a constant  for any given set $(\ell_1,\ldots,\ell_{|\Lambda|})$ with $\sum_j\ell_j=K$, in particular,
\begin{equation}\label{sum-alphas}
\sum_{{\tiny\begin{array}{c}
        \alpha_1=\ell_1,\ldots,\alpha_{|\Lambda|}=
        \ell_{|\Lambda|} \\
        \alpha_1+\ldots+\alpha_{|\Lambda|}=N\\
      \end{array}}}^N\ \
\prod_{j=1}^{|\Lambda|}\binom{\alpha_j}{\ell_j}
\chi_{\ell_1+\ldots+\ell_{|\Lambda|}=K}=
\binom{N+|\Lambda|-1}{K+|\Lambda|-1}.
\end{equation}
Which follows from the use of the identity
\begin{equation}
\binom{n+1}{k+1}=\sum_{\alpha=\ell}^{n}
\binom{\alpha}{\ell}\binom{n-\alpha}{k-\ell}, \quad \text{for}\quad 0\leq\ell\leq k.
\end{equation}
Finally, we use the fact that for any $m$ and $n$ in $\mathbb{N}$,
\begin{equation}
(x_1+x_2+\ldots+x_m)^n=\sum_{k_1+k_2+\ldots+k_m=n}\frac{n!}{k_1!\ k_2!\ldots k_m!}x_1^{k_1}x_2^{k_2}\ldots x_m^{k_m}
\end{equation}
we conclude that the sum of all terms of degree $0\leq K\leq N$ from the left hand side of (\ref{eq:I-N}) can be written as
\begin{equation}
a_K\left(\sum_{j=1}^{|\Lambda|}x_j\right)^K, \text{\, where \, }a_K:=(-1)^K\binom{N+|\Lambda|-1}{K+|\Lambda|-1}\frac{1}{K!}.
\end{equation}
This completes the proof of the lemma with the polynomial $\mathcal{Q}_N(x):=\sum_{j=0}^N a_j x^j$.


\begin{thebibliography}{99}







%%%%%%%%%%%%%%%%%%%
%%%%%%%%%%%%%%%%%%%
%%%%%%%%%%%%%%%%%%%



\bibitem{AbanicPepic17} D. A. Abanin and Z. Papi\'{c}, \emph{Recent progress in many-body localization},  Ann. Phys. (Berlin) \textbf{529} (2017), 1700169



\bibitem{ANSS17} H. Abdul-Rahman, B. Nachtergaele, R. Sims and G. Stolz, \emph{Localization properties of the XY spin chain: a review of mathematical results with an eye toward many-body localization}, Ann. Phys. (Berlin) \textbf{529}, 1600280


\bibitem{ANSS16} H. Abdul-Rahman, B. Nachtergaele, R. Sims, and G. Stolz, \emph{Entanglement dynamics of disordered quantum XY chains}, Lett. Math. Phys. \textbf{106} (2016), 649--674

\bibitem{ARSS17} H. Abdul-Rahman, R. Sims, and G. Stolz, \emph{Correlations in disordered quantum harmonic oscillator systems: The effects of excitations and quantum quenches}, arXiv:1704.04841 (2017)

\bibitem{ARS15} H. Abdul-Rahman and G. Stolz,  \emph{A uniform area law for the entanglement of eigenstates in the disordered XY chain},  J. Math. Phys. \textbf{56} (2015), 121901, 25 pp.



\bibitem{Agarwaletal} K. Agarwal, E. Altman, E. Demler, S. Gopalakrishnan, D. A. Huse, and M. Knap, \emph{Rare-region effects and dynamics near the many-body localization transition},  Ann. Phys. (Berlin) \textbf{529} (2017), 1600326

\bibitem{Altmam-etal-15} E. Altman and R. Vosk, \emph{Universal dynamics and renormalization in many body localized systems}, Annu. Rev. Condens. Matter Phys. \textbf{6} (2015), 383--409

\bibitem{Araki71} H. Araki and M. Shiraishi, \emph{On quasifree states of the canonical
commutation relations (I)}, Publ. RIMS, Kyoto Univ. \textbf{7} (1971/72), 105--120


\bibitem{Audenaert2002}
K. Audenaert, J. Eisert, and M. B. Plenio,
\emph{Entanglement properties of the harmonic chain},
Phys. Rev. A \textbf{66} (2002), 042327


\bibitem{Bartlett02} S. D. Bartlett, B. C. Sanders, S. L. Braunstein, and K. Nemoto, \emph{Efficient classical simulation of continuous variable quantum information processes}, Phys. Rev. Lett. \textbf{88} (2002), 097904

\bibitem{BW18} V. Beaud and S. Warzel, \emph{Bounds on the entanglement entropy of droplet states in the XXZ spin chain},
J. Math. Phys. \textbf{59} (2018), 012109

\bibitem{Benzietal15} M. Benzi and V. Simoncini, \emph{Decay bounds for functions of matrices with banded or Kronecker structure}, SIAM J. Matrix Anal. Appl. \textbf{36} (2015), 1263--1282

\bibitem{BH13} F. Brandao and M. Horodecki, \emph{An area law for entanglement from exponential decay of correlations},
Nature Physics \textbf{9} (2013), 721--726

\bibitem{BH} F. Brandao and M. Horodecki, \emph{Exponential decay of correlations implies area law}, Commun. Math. Phys. \textbf{333} (2015), 761--798

\bibitem{BRvol2} O. Bratteli and D. Robinson, \emph{Operator algebras and quantum statistical mechanics 2}, 2nd ed., New York, NY, Springer Verlag, 1997


\bibitem{BD07} L. Bruneau and J. Derezi\'{n}ski, \emph{Bogoliubov Hamiltonians and one-parameter groups of Bogoliubov transformations}, J. Math. Phys. \textbf{48} (2007), 022101


\bibitem{Dell07} F. Dell'Anno, S. De Siena, L. Albano Farias, and F. Illuminati, \emph{Continuous variable quantum teleportation with non-Gaussian resources}, Phys. Rev. A \textbf{76} (2007), 022301


\bibitem{Dong08} R. Dong, M. Lassen, J. Heersink, C. Marquardt, R. Filip, G. Leuchs, and U. L. Andersen, \emph{Experimental entanglement distillation of mesoscopic quantum states}, Nature Physics \textbf{4} (2008), 919--923


\bibitem{EisertCramerPlenio2010} J. Eisert,
M. Cramer and M. B. Plenio, \emph{Area laws for the entanglement entropy}, Rev. Mod. Phys. \textbf{82} (2010), 277


\bibitem{Eisert02} J. Eisert, S. Scheel, and M. B. Plenio, \emph{Distilling Gaussian states with Gaussian operations is impossible},
Phys. Rev. Lett. \textbf{89} (2002), 137903


\bibitem{EKS} A. Elgart, A. Klein, and G. Stolz, \emph{Many-body localization in the droplet spectrum of the random XXZ quantum spin chain}, J. Funct. Anal. (2018) https://doi.org/10.1016/j.jfa.2017.11.001


\bibitem{Feller50} W. Feller, \emph{An introduction to probability theory and its applications}, 2nd ed., Wiley, 1950

%\bibitem{FS17} C. Fischbacher and G. Stolz, \emph{Droplet states in quantum XXZ spin systems on general graphs},  arXiv:1712.10276

\bibitem{MES-Eisert} C. Gogolin, M. P.  Mueller, and J. Eisert, \emph{Absence of thermalization in nonintegrable systems},
Phys. Rev. Lett. \textbf{106} (2011), 040401

\bibitem{Gomesetal09} R. M. Gomes, A. Salles, F. Toscano, P. H. Souto Ribeiro, and S. P. Walborn, \emph{Quantum entanglement beyond Gaussian criteria}, Proc. Natl. Acad. Sci. U.S.A., \textbf{106} (2009), 21517--21520

\bibitem{deGosson} M. de Gosson, \emph{Symplectic geometry and quantum mechanics}, Birkh\"{a}user, Basel, series ``Operator Theory: Advances and Applications'', 2006


\bibitem{Hage08} B. Hage, A. Samblowski, J. DiGuglielmo, A. Franzen, J. Fiur\'{a}\u{s}ek, and R. Schnabel, \emph{Preparation of distilled and purified continuous-variable entangled states}, Nature Physics \textbf{4} (2008), 919--923


\bibitem{HorodeckiPPT2} M.  Horodecki, P. Horodecki, and R. Horodecki, \emph{Separability of mixed states: necessary and sufficient conditions}, Phys. Lett. A \textbf{223} (1996), 1--8


\bibitem{Horodecki09} R. Horodecki, P. Horodecki, M. Horodecki, and K. Horodecki, \emph{Quantum entanglement}, Rev. Mod. Phys. \textbf{81} (2009), 865--942




\bibitem{LamiAdessoetal17} L. Lami, A. Serafini, and G. Adesso, \emph{Gaussian entanglement revisited},  New J. Phys. \textbf{20} (2018) 023030

\bibitem{Man68} J. Manuceau and A. Verbeure, \emph{Quasi-free states of the
CCR algebra and Bogoliubov transformations}, Commun. Math. Phys. \textbf{9} (1968), 293--302


\bibitem{Messiah}  A. Messiah, \emph{Quantum mechanics}, Dover, New York, 1999


\bibitem{NSS1} B. Nachtergaele, R. Sims, and G. Stolz, \emph{Quantum harmonic oscillator systems with disorder}, J. Stat. Phys. \textbf{149} (2012), 969--1012

\bibitem{NSS2}  B. Nachtergaele, R. Sims, and G. Stolz, \emph{An area law for the bipartite entanglement of disordered oscillator systems}, J. Math. Phys. \textbf{54} (2013), 042110

\bibitem{NH} R. Nandkishore and D. A. Huse, \emph{Many body localization and thermalization in quantum statistical mechanics}, Annu. Rev. Condens. Matter Phys. \textbf{6} (2015), 15--38

\bibitem{Nielsen2000} M. Nielsen and I. Chuang, \emph{Quantum computation and quantum information}
Cambridge University Press, 2000

\bibitem{Opartny04} T. Opatrny, G. Kurizki, and D.-G. Welsch, \emph{Continuous-variable teleportation improvement by photon subtraction via conditional measurement}, Phys. Rev. A \textbf{61} (2000) 032302

\bibitem{PeresPPT1} A. Peres, \emph{Separability criterion for density matrices}, Phys. Rev. Lett. \textbf{77} (1996), 1413

\bibitem{ReedSimon} M. Reed, B. Simon, \emph{Methods of modern mathematical physics}, Academic Press, San Diego, 1975, Vol. 2


\bibitem{SCW} N. Schuch, J.I. Cirac, and M. Wolf,
\emph{Quantum states on harmonic lattices},
Commun. Math. Phys. \textbf{267} (2006), 65--95

\bibitem{SW} R. Seiringer and S. Warzel, \emph{Decay of correlations and absence of superfluidity in the disordered Tonks-Girardeau gas}, New J. Phys. \textbf{18} (2016), 035002

\bibitem{Kaushik15} K. Seshadreesan, J. Dowling, and G.  Agarwal,
\emph{Non-Gaussian entangled states and quantum teleportation of {S}chr\"{o}dinger-cat states}, Phys. Scripta \textbf{90} (2015), 074029

\bibitem{Shale62} D. Shale, \emph{Linear symmetries of free boson fields}, Trans. Am. Math. Soc. \textbf{103} (1962), 149--167

\bibitem{SimsWarzel} R. Sims and S. Warzel, \emph{Decay of determinantal and pfaffian correlation functionals in one-dimensional lattices}, Commun. Math. Phys. \textbf{347} (2016), 903--931


\bibitem{Stanley97} R. Stanley, \emph{Enumerative combinatorics}, Cambridge University Press, 1997

\bibitem{Strobeletal14} H. Strobel, W. Muessel, D. Linnemann, T. Zibold, D. B. Hume, L. Pezz\`{e}, A. Smerzi, and M. K. Oberthaler, \emph{Fisher information and entanglement of non-Gaussian spin states}, Science \textbf{345} (2014), 424--427

\bibitem{Truesdell} C. Truesdell, \emph{On the addition and multiplication theorems for the special functions}, Proc. Natl. Acad. Sci. U.S.A., \textbf{36} (1950), 752-757



\bibitem{VidalWerner} G. Vidal and R. Werner, \emph{Computable measure of entanglement},
Phys. Rev. A \textbf{65} (2002), 032314


\bibitem{Walschaerseta17} M. Walschaers, C. Fabre, V. Parigi, and N. Treps, \emph{Entanglement and Wigner function negativity of multimode non-Gaussian states}, Physical Review Letters \textbf{119} (2017), 183601

\bibitem{Wangetal08} X.-B. Wang, T. Hiroshima, A. Tomita, and M. Hayashi, \emph{Quantum information with Gaussian states}, Physics Reports \textbf{448} (2007), 1--111

\bibitem{Weedbrook12} C. Weedbrook, S. Pirandola, R. Garc\'{i}a-Patr\'{o}n, N. J. Cerf, T. C. Ralph, J. H. Shapiro, and S. Lloyd, \emph{Gaussian
quantum information}, Rev. Mod. Phys. \textbf{84} (2012), 621

\bibitem{Weidmann} J. Weidmann, \emph{Linear operators in Hilbert spaces}, Graduate Texts in Mathematics, Vol \textbf{68}, Springer, New York, 1980




\end{thebibliography}
\end{document}